\documentclass[11pt,intlimits,reqno]{amsart}
\usepackage{amssymb}
\usepackage{amsthm}
\usepackage{enumerate}
\usepackage{verbatim}
\usepackage{bbm}

\newtheorem{thm}{Theorem}[section]
\newtheorem{prop}[thm]{Proposition}

\newcommand{\pair}[2]{\langle #1\;;\;#2 \rangle}

\newcommand{\R}{\mathbb R}

\newcommand{\N}{\mathbb N}
\newcommand{\so}{\mathfrak{so}}
\newcommand{\pb}{\{\,\cdot\,,\,\cdot\,\}}
\newcommand{\cmm}{[\,\cdot\,,\,\cdot\,]}
\renewcommand{\to}{\rightarrow}
\DeclareMathOperator{\Tr}{Tr}
\DeclareMathOperator{\ad}{ad}
\renewcommand{\S}{\mathcal S}
\newcommand{\A}{\mathcal{A}}
\newcommand{\1}{\mathbbm 1}
\newcommand{\0}{\mathbbm O}
\newcommand{\bra}[1]{\left\langle #1 \right|}
\newcommand{\ket}[1]{\left| #1 \right\rangle}

\newcommand{\Aa}{\A_{a_1,\ldots,a_{n-1}}}
\newcommand{\Sa}{\S_{a_1,\ldots,a_{n-1}}}

\newcommand{\be}{\begin{equation}}
\newcommand{\ee}{\end{equation}}
\newcommand{\ben}{\begin{enumerate}}
\newcommand{\een}{\end{enumerate}}
\newcommand{\bit}{\begin{itemize}}
\newcommand{\eit}{\end{itemize}}
\begin{document}
\title[Lie bundle on the space of deformed$\ldots$]{Lie bundle on the space of deformed skew-symmetric matrices}
\author{Alina Dobrogowska, Tomasz Goli\'nski}
\dedicatory{University in Bia{\l}ystok\\
Institute of Mathematics\\
Lipowa 41, 15-424 Bia{\l}ystok, Poland\\
email: alaryzko@alpha.uwb.edu.pl, tomaszg@alpha.uwb.edu.pl
}

\begin{abstract}
We study a Lie algebra $\Aa$ of deformed skew-symmetric $n\times n$ matrices endowed with a Lie bracket given by a choice of deformed symmetric matrix. The deformations are parametrized by a sequence of real numbers $a_1,\ldots,a_{n-1}$. Using isomorphism $(\Aa)^*\cong L_+$ we introduce a Lie-Poisson structure on the space of upper-triangular matrices $L_+$. In this way we generate hierarchies of Hamilton systems with bihamiltonian structure. 
\end{abstract}

\maketitle
\tableofcontents

\section{Introduction}

A \emph{Lie bundle} on a finite dimensional vector space $V$ (also called by some authors a \emph{linear bundle of Lie algebras}) is a family of compatible Lie structures on $V$, i.e. a family of Lie brackets $\cmm_S:V\times V\to V$ defining Lie algebra structure on $V$ depending linearly on the parameter $S$ belonging to some real vector space $W$
\be \label{Lb-lin}\cmm_{aS+bT}=a\cmm_{S}+b\cmm_{T}\qquad S,T\in W,\; a,b\in\R.\ee
The notion was studied and applied mostly in the context of integrable systems, see \cite{yan,fomenko}.

Given a Lie bundle structure on $V$ there exists a family of Lie--Poisson brackets $\pb_S$ on dual space $V^*$
\be \{f,g\}_S(x)=\pair{x}{[Df(x),Dg(x)]_S},\qquad x\in V^*, Df(x), Dg(x)\in (V^*)^*\cong V.\ee
The condition \eqref{Lb-lin} ensures that these Poisson brackets are compatible since the following condition holds
\be \pb_{aS+bT}=a\pb_S+b\pb_T,\qquad S,T\in W, a,b\in\R.\ee

The simplest but still interesting case is when $W$ is two dimensional. In that case we have 
a one-parameter family of Poisson brackets (up to the constant factor) called a \emph{Poisson pencil}. It is well known that given that 
Poisson brackets from that family are not proportional, Casimir functions for 
one of them are in involution with respect to any Poisson bracket from the 
pencil, see \cite{bols-bor}. It is a version of Magri method \cite{magri} and it 
allows to construct bihamiltonian systems which are usually much easier to 
integrate due to the presence of many additional integrals of motion.

A classical example of such a structure can be obtained by considering $V$ as a space of skew-symmetric $n\times n$ matrices $\so(n)$. The natural pairing $\pair\rho X=-\Tr \rho X$, $\rho\in \so(n)^*$, $X\in \so(n)$, identifies dual space $\so(n)^*$ with $\so(n)$. On that space one considers a family of Lie brackets given by
\be \label{ex-lb}[X,Y]_S=XSY-YSX,\ee
where $S$ is a symmetric matrix, see e.g. \cite{Morosi,yan,gerdjikov,borisov-mamaev}. We will denote these Lie algebras as $\so_S(n)$. In case $S=\1$ we obtain the usual $\so(n)$. One can also consider \eqref{ex-lb} as a Lie bracket on the space of symmetric matrices by taking $S$ to be an antisymmetric matrix, see e.g. \cite{Ratiu-Bloch}.

In this paper, following \cite{DO-L2}, we will replace the Lie algebra $\so(n)$ with a set $\Aa$ of deformed skew-symmetric matrices (see \eqref{Aa}) given by a finite sequence of real parameters $a_1, \ldots, a_{n-1}$. Moreover by $\Sa$ we will denote a set of deformed symmetric matrices (see \eqref{Sa}). Both these sets are vector spaces and on $\Aa$ we define a family of Lie brackets by the formula \eqref{ex-lb} for $S\in\Sa$. In this way we obtain a Lie bundle on $\Aa$. The dual space $(\Aa)^*$ will be identified with a set of upper-triangular matrices $L_+$.

The Lie bundle obtained in this way can be decomposed into simpler bundles using classification theorem by I.L. Cantor and D.E. Persits, see \cite{cantor-persits,fomenko,yan}. Some results in this direction are presented in Section~3. However we are more interested in arising Hamilton equations than in the structure of the bundle itself.

Let us note that Lie algebra $\Aa$ in a generic case (all $a_i\neq0$ and $\det S\neq 0$) is isomorphic to $\so(p,q)$ (see Proposition \ref{prop:so_pq}). However in particular cases we obtain contractions of Lie algebras \cite{wigner-contr}. 
For example if we put $a_1=0$, $a_2=\ldots=a_{n-1}=1$, we obtain Lie algebra $\mathfrak e(n-1)$ of Euclidean group with Lie bracket given as a commutator modified by a matrix $S$. More examples of integrable systems related to Lie algebra $\Aa$ with standard commutator were considered in \cite{alina-relativistic,alina-ratiu,alina-sigma}.

Section~2 introduces the Lie algebra $\Aa$ of deformed skew symmetric $n\times n$ matrices and the vector space of deformed symmetric matrices $\Sa$. They are both defined by a fixed sequence of parameters $a_1,\ldots, a_{n-1}$. We have proven that the formula \eqref{ex-lb} defines a Lie bracket on $\Aa$ given that $S\in \Sa$. We have also described $\Aa$ and $\Sa$ in block form and proven several useful facts about these sets.

In Section~3 we describe certain isomorphisms between Lie algebras $\Aa$ for various choices of parameters $a_1,\ldots, a_{n-1}$ and $S\in\Sa$. In particular isomorphisms with the case $a_i=\pm1$, $i=1,\ldots,n-1$, are established.

Section~4 deals with the Lie--Poisson structure of the dual space to $\Aa$. We chose to represent $(\Aa)^*$ as the space of upper-triangular matrices $L_+$ with pairing given by trace. We also present the formula for Casimirs for that Poisson space in the generic case. By considering a Poisson pencil we are able to introduce a hierarchy of Hamilton equations.

In Section~5 we present some examples of Hamilton systems in the case $n=4$. Explicit formulas for the Poisson bracket and the Casimirs are given. Using vector notation we are able to present certain Hamilton equation in a compact form both in the generic form and in one of the  contractions. We also observe their relationship with classical integrable systems such as Clebsch system or multidimensional rigid body.

\section{Lie algebra parametrized by a finite sequence of numbers}
For a fixed sequence of real numbers $a_1,\ldots, a_{n-1}$ we define the following sets of matrices
\be\label{Aa}\Aa:=\{X=(x_{ij})\in Mat_{n\times n}(\R)\;|\;
x_{ij}=-a_i\cdots a_{j-1}x_{ji}\textrm{ for }j>i, x_{ii}=0\}\ee
and
\be\label{Sa}\Sa:=\{S=(s_{ij})\in Mat_{n\times n}(\R)\;|\;
s_{ij}=a_i\cdots a_{j-1}s_{ji}\textrm{ for }j>i\}.\ee
Namely these sets contain matrices of the following forms
\begin{equation}
X=\left(
\begin{array}{cccccc}
0       & -a_{1}x_{21} & -a_{1}a_2x_{31} &  -a_{1}a_{2}a_{3}x_{41} &\ldots& -a_1a_2\cdots a_{n-1}x_{n1} \\
x_{21}  & 0            & -a_{2}x_{32} & -a_{2}a_{3}x_{42} & \ldots & -a_2a_3\cdots a_{n-1}x_{n2}\\
x_{31}  & x_{32}      & 0        & -a_{3}x_{43}    & \ldots & -a_3a_4\cdots a_{n-1}x_{n3}\\
x_{41}  & x_{42}      & x_{43}   &  0    & \ldots & -a_4a_5\cdots a_{n-1}x_{n4}\\
\vdots  & \vdots       & \vdots  &\vdots     & \ddots& \vdots\\
x_{n1} & x_{n2} 		&x_{n3}	 & x_{n4}	& \ldots	& 0
\end{array}\right),
\end{equation}
\begin{equation}
S=\left(
\begin{array}{cccccc}
s_{11}       & a_{1}s_{21} & a_{1}a_2s_{31} &  a_{1}a_{2}a_{3}s_{41} &\ldots& a_1a_2\cdots a_{n-1}s_{n1} \\
s_{21}  & s_{22}            & a_{2}s_{32} & a_{2}a_{3}s_{42} & \ldots & a_2a_3\cdots a_{n-1}s_{n2}\\
s_{31}  & s_{32}      & s_{33}        & a_{3}s_{43}    & \ldots & a_3a_4\cdots a_{n-1}s_{n3}\\
s_{41}  & s_{42}      & s_{43}   &  s_{44}    & \ldots & a_4a_5\cdots a_{n-1}s_{n4}\\
\vdots  & \vdots       & \vdots  &\vdots     & \ddots& \vdots\\
s_{n1} & s_{n2} 		&s_{n3}	 & s_{n4}	& \ldots	& s_{nn}
\end{array}\right).
\end{equation}

The sets $\Aa$ and $\Sa$ can be viewed as deformations of the sets of antisymmetric and symmetric matrices respectively. They were introduced in \cite{DO-L2} and it was shown that $\Aa$ is a Lie algebra with respect to the standard commutator. In the generic case this deformation is trivial but special cases give rise to Lie algebra contractions. In this paper however we are going to investigate its Lie algebra structure with respect to a deformed Lie bracket of the form
\begin{equation}\label{Sbracket}
[X,Y]_S:=XSY-YSX
\end{equation}
for $X,Y\in\Aa$ and a fixed $S\in\Sa$. By considering $\Aa$ with a family of Lie brackets $\cmm_S$ for $S\in\Sa$ we will obtain a Lie bundle.

Before we prove that \eqref{Sbracket} defines a Lie algebra structure on $\Aa$, let us 
describe in a more direct way the sets $\Aa$ and $\Sa$.
If all parameters $a_1,\ldots, a_{n-1}$ are non-zero it is easy to observe that
\begin{equation}\label{delta-X}
X\in\Aa \Longleftrightarrow  \delta X+X^T\delta=0,\end{equation}
\begin{equation}\label{delta-S}
S\in\Sa \Longleftrightarrow  \delta S-S^T\delta=0,\end{equation}
where $\delta$ is the following diagonal matrix (in the canonical basis $\{\ket i\}_{i=1}^n$)
\be\delta:=
\ket 1\bra 1+ \sum_{i=2}^{n} a_{1}a_{2}\cdots a_{i-1} \ket i \bra i=
\left(
\begin{array}{cccccc}
1 & 0 & 0 & \ldots & 0 \\
0 & a_1 & 0 & \ldots & 0\\
0 & 0 & a_1a_2&\ldots & 0 \\
\vdots& \vdots& \vdots & \ddots & \vdots\\
0 & 0 &0 & \ldots & a_1a_2\cdots a_{n-1}
\end{array}\right).
\ee

In order to investigate the situation when one or more parameters is equal to zero, we introduce the following notation. Let $P_l$ be a projector onto $l+1,\ldots, n$ vectors of the basis:
\be P_l=\sum_{i=l+1}^n \ket i \bra i\ee
and by $\delta_{l}$ we will denote the matrix
\be \delta_{l}:=\ket {l+1}\bra {l+1}+ \sum_{i=l+2}^{n} a_{l+1}a_{l+2}\cdots a_{i-1} \ket i \bra i,\ee
which has the following block form
\be\delta_{l}=\left(
\begin{array}{c|ccccc}
\0 &  && \0 &&  \\
\hline
 & 1 & 0 & 0 & \ldots & 0\\
 & 0 & a_{l+1}& 0 &\ldots & 0 \\
\0 & 0& 0 & a_{l+1}a_{l+2}& \ldots&0\\
&\vdots&\vdots&\vdots&\ddots& \vdots\\
&0 & 0 & 0 &\ldots&a_{l+1}a_{l+2}\cdots a_{n-1}
\end{array}\right).\ee

Now if one of the parameters is equal to zero $a_k=0$ it is easy to calculate that condition \eqref{delta-X} should be replaced by the following two conditions:
\be \label{delta-X2}
\left\{\begin{array}{l}
\delta X + X^T\delta=0,\\
\delta_{k} X P_k + P_k X^T \delta_{k}=0.
\end{array}\right.\ee
In the block form this condition means that $X\in\Aa$ is of the following form
\be \label{block-tr}
X=\left(
\begin{array}{c|c}
A &  \0   \\
\hline
B & C \\
\end{array}\right),\ee
where $A\in \A_{a_1,\ldots,a_{k-1}}$, $C\in \A_{a_{k+1},\ldots,a_{n-1}}$ and $B$ is an arbitrary matrix.

Let us now consider the most general case with $N$ separate parameters equal to zero, namely
\be a_{k_1}=\ldots=a_{k_N}=0,\ee
where the sequence $(k_1,\ldots, k_N)$ is chosen to be strictly increasing. For consistency we put additionally $k_0=0$ and $k_{N+1}=n$. 
Then the condition for matrix $X$ to belong to $\Aa$ takes the form of the following system of equations:
\be\label{delta-XX}\delta_{k_i} X P_{k_i} + P_{k_i} X^T \delta_{k_i}=0,\qquad i=0,1,2,\ldots,N.\ee
By analogy condition for a matrix $S$ to belong to $\Sa$ assumes the form
\be\label{delta-SS}\delta_{k_i} S P_{k_i} - P_{k_i} S^T \delta_{k_i}=0,\qquad i=0,1,2,\ldots,N.\ee

Note that in this notation $\delta_0\equiv \delta_{k_0}\equiv \delta$, $P_0\equiv P_{k_0}\equiv \1$, 
and $\delta_{k_{N+1}}=\delta_n=P_{k_{N+1}}=P_n=0$. 
So conditions \eqref{delta-X} and \eqref{delta-S} are included in these systems of equations for $i=0$.

We obtain the following relations between introduced matrices:
\be\delta_{k_i} P_{k_j}=0, \qquad P_{k_i} P_{k_j}=P_{k_j}\ee
for $j>i$ and
\be\delta_{k_i} \delta_{k_j}=0\ee
for $i\neq j$. Moreover all $P_{k_i}$ and $\delta_{k_j}$ commute and
 $\delta_{k_i}$ are invertible when restricted to the range of projectors $P_{k_{i}}-P_{k_{i+1}}$. We will denote that pseudoinverse elements by $\iota(\delta_{k_i})$:
 \be \label{iota}\iota(\delta_{k_i})\delta_{k_i}=\delta_{k_i}\iota(\delta_{k_i})=P_{k_{i}}-P_{k_{i+1}}.\ee

From this we can establish the following generalization of the block representation \eqref{block-tr}:

\begin{prop}\label{rem-schodki}
For $X\in \Aa$ and $S\in \Sa$ the following relations hold
\be \label{schodki}(\1-P_{k_i}) X P_{k_i}=(\1-P_{k_i}) S P_{k_i}=0, \ee
\be \label{schodki2}P_{k_i} X P_{k_i} = XP_{k_i},\qquad P_{k_i} S P_{k_i} = SP_{k_i}\ee
for $i=0,\ldots, N$.
\end{prop}
\begin{proof}We begin by proving this proposition for $X\in\Aa$. 
Let us take the relation \eqref{delta-XX} for $k_j$, $j<i$. If we multiply it from the right by $P_{k_i}$ we obtain
$$\delta_{k_j} X P_{k_i}=0.$$
By multiplying from the left by $\iota(\delta_{k_j})$ we obtain 
$$(P_{k_j}-P_{k_{j+1}})XP_{k_i}=0.$$
To obtain the equality \eqref{schodki} we sum up the above equalities from $j=0$ to $j=i-1$.  The second equality follows trivially.

The proof for $S\in\Sa$ is completely analogous.
\end{proof}

From \eqref{schodki} it follows that $X\in\Aa$ can be represented in the following block form
\be \label{block-schodki}
X=\left(
\begin{array}{c|c|c|c}
A_0 &  \0  & \cdots & \0 \\
\hline
* & A_1 &\cdots& \0 \\
\hline
\vdots & \vdots & \ddots & \vdots \\
\hline
* & * & \cdots& A_N \\
\end{array}\right),\ee
where $A_i\in \A_{a_{k_i+1},\ldots,a_{k_{i+1}-1}}$, $i=0,\ldots N$, and $*$ denotes arbitrary matrices of suitable sizes.

\begin{prop}For a fixed $S\in\Sa$ the set $\Aa$ is a Lie algebra with respect to the bracket \eqref{Sbracket}.
\end{prop}
\begin{proof}The only non-trivial thing to check is
that $\Aa$ is closed under bracket \eqref{Sbracket}, namely if
$[X,Y]_S$ satisfies relations \eqref{delta-XX}
$$\delta_{k_i} (XSY-YSX) P_{k_i} + P_{k_i} (XSY-YSX)^T\delta_{k_i}=0$$
for $X,Y\in \Aa$, $S\in \Sa$ and $i=0,1,\ldots, N$.
Expanding left hand side we get
$$\delta_{k_i} XSYP_{k_i}-\delta_{k_i} YSXP_{k_i} + P_{k_i}Y^TS^TX^T\delta_{k_i}-P_{k_i}X^TS^TY^T\delta_{k_i}=0$$
Using Proposition \ref{rem-schodki} and the relations $P_{k_i}X^T\delta_{k_i}=-\delta_{k_i} XP_{k_i}$, $P_{k_i}Y^T\delta_{k_i}=-\delta YP_{k_i}$ and $P_{k_i}S^T\delta_{k_i}=\delta_{k_i} SP_{k_i}$ we see that
$$P_{k_i}Y^TS^TX^T\delta_{k_i}=P_{k_i}Y^TP_{k_i}S^TP_{k_i}X^T\delta_{k_i}= \delta_{k_i} Y P_{k_i} S P_{k_i} X P_{k_i} = \delta_{k_i}YSXP_{k_i}.$$
Performing the same computation in the other term we see that all terms in the relation cancel.

%

\end{proof}

Let us prove several basic properties which will be useful in the sequel.

\begin{prop}\label{prop:5}
If matrix $S\in\Sa$ or $X\in\Aa$ is invertible then its inverse also belongs to the same set, i.e. $S^{-1}\in\Sa$ or $X^{-1}\in\Aa$.
 \end{prop}
\begin{proof}
Let us prove the first assertion in the case when all $a_i\neq 0$. It is equivalent to
$$ \delta S^{-1} - (S^T)^{-1}\delta=0.$$
By multiplying this condition from the right by $S$ and using \eqref{delta-S} we get
$$ (\delta S^{-1} - (S^T)^{-1}\delta)S=\delta- (S^T)^{-1}\delta S = \delta-\delta=0.$$

Now if we admit the case when one of the parameters $a_k$ vanishes, we can write $S$
in the block form
$$ S=\left(
\begin{array}{c|c}
\tilde A &  \0   \\
\hline
\tilde B & \tilde C \\
\end{array}\right),$$
where $\tilde A\in \S_{a_1,\ldots,a_{k-1}}$, $\tilde C\in \S_{a_{k+1},\ldots,a_{n-1}}$ and $\tilde B$ is an arbitrary matrix. If $S$ is invertible then so are $\tilde A$ and $\tilde C$ and we have
$$ S^{-1}=\left(
\begin{array}{c|c}
\tilde A^{-1} &  \0   \\
\hline
-\tilde C^{-1}\tilde B\tilde A^{-1} & \tilde C^{-1} \\
\end{array}\right).$$

We already proved that $\tilde A^{-1}\in \S_{a_1,\ldots,a_{k-1}}$ and $\tilde C^{-1}\in \S_{a_{k+1},\ldots,a_{n-1}}$. Thus we gather that $S^{-1}\in \Sa$.

The same argument works also in the case of more parameters vanishing, but for the sake of completeness let us present the exact calculation completing the proof in the most general case. 

First of all, let us notice that equalities \eqref{schodki}-\eqref{schodki2} are valid also for $S^{-1}$. To prove that let us decompose equality $SS^{-1}=S^{-1}S=\1$ with respect to a pair of complementary projectors $P_{k_i}$ and $\1-P_{k_i}$. Among others we obtain the following equalities
$$(\1-P_{k_i}) S^{-1} (\1-P_{k_i}) S (\1-P_{k_i})=\1-P_{k_i}$$
$$(\1-P_{k_i}) S (\1-P_{k_i}) S^{-1} P_{k_i}=0.$$
By multiplying the second one by $(\1-P_{k_i}) S^{-1}$ and using the first one, we get the desired result.

The proposition we are to prove is equivalent to the set of equalities 
$$ \delta_{k_i} S^{-1} P_{k_i}  = P_{k_i} (S^{-1})^T \delta_{k_i}.$$
By multiplying this condition from the right side by $P_{k_i} S P_{k_i}$
and taking into account equality \eqref{delta-SS} and that $P_{k_i}\delta_{k_i}=\delta_{k_i}$ we obtain
$$\delta_{k_i} P_{k_i} S^{-1}P_{k_i} S P_{k_i} =  (P_{k_i} S^{-1}P_{k_i})^T S^T\delta_{k_i}.$$
Taking into account relation \eqref{schodki} we obtain identity.
To conclude the proof we note that $P_{k_i} S P_{k_i}$ is invertible when restricted to the range of $P_{k_i}$ so we didn't lose generality by multiplying by it.

The claim that $X^{-1}\in\Aa$ follows accordingly.
\end{proof}

\begin{prop}\label{prop:2}
For $S\in\Sa$ and $X\in\Aa$ we have
\be \label{SXS}S(XS)^{2k-1}, X(SX)^{2k}\in\Aa,\ee
\be \label{XSX}X(SX)^{2k-1}, S(XS)^{2k}\in \Sa,\ee
for $k\in\N$.
\end{prop}
\begin{proof}
For simplicity let us prove only one of these statements in the case $k=1$ and $a_i\neq 0$. Proof in a more general case is analogous.

The computation to check that $SXS\in \Aa$ is straightforward
$$ \delta(SXS)+(SXS)^T\delta=\delta SXS+S^TX^TS^T\delta=
\delta SXS+S^TX^T\delta S=$$
$$=\delta SXS-S^T\delta XS=\delta SXS - \delta SXS=0.$$
To check remaining statements we make completely analogous sequence of transformations.

\end{proof}

\begin{prop}\label{prop:3}
 For $X\in\Aa$ and $S\in\Sa$ we have
\be \Tr SX=0.\ee
\end{prop}
\begin{proof}
Note that for $a_i\neq 0$, from \eqref{delta-X}-\eqref{delta-S} it follows that
$$ \delta SX+(XS)^T\delta=0.$$
Thus by properties of trace we get
$$ \Tr SX= -\Tr\delta^{-1}(XS)^T\delta = -\Tr XS= - \Tr SX.$$

Now for the general case, we note that both matrix $S$ and $X$ have block representation with diagonal blocks belonging to lower-dimensional $\A$ and $\S$ spaces. To conclude the proof it is sufficient to observe that while calculating $\Tr SX$ only diagonal blocks (one from $\A$ and the other from $\S$) get paired together. Namely from Proposition \ref{rem-schodki} we obtain
$$\Tr SX = \sum_{i=0}^N\Tr (P_{k_{i}}-P_{k_{i+1}})S(P_{k_{i}}-P_{k_{i+1}})X(P_{k_{i}}-P_{k_{i+1}}).$$
From \eqref{delta-XX} and \eqref{delta-SS} (or from \eqref{block-schodki}) it follows that each term of that sum vanishes.
\end{proof}

%
%

\begin{prop}\label{prop:4}
For any matrix $A\in Mat_{n\times n}(\R)$ the following matrix
\be A-\sum_{i=0}^N\iota(\delta_{k_i})A^T\delta_{k_i}-\sum_{i=1}^N(P_{k_{i-1}}-P_{k_i})AP_{k_i}\ee
belongs to the Lie algebra $\Aa$.

In the case when all $a_i\neq 0$ then that expression simplifies to
\be A-\delta^{-1}A^T\delta\in \Aa.\ee
\end{prop}
\begin{proof}
Let us denote the considered matrix by $X$. We have to check the conditions \eqref{delta-XX}. Direct calculation using properties of $\delta_{k_i}$ and $P_{k_i}$ yields
$$\delta_{k_l}X P_{k_l} = \delta_{k_l}A P_{k_l} - \delta_{k_l}\iota(\delta_{k_l})A^T\delta_{k_l}-\delta_{k_l}AP_{k_{l+1}},$$
$$P_{k_l}X^T\delta_{k_l}=P_{k_l}A^T\delta_{k_l}-\delta_{k_l}A\iota(\delta_{k_l})\delta_{k_l}-P_{k_{l+1}}A^T\delta_{k_l},$$
for any $l=0,1,\ldots,N$.
Taking into account \eqref{iota} and adding up both equalities we see that all terms cancel and thus \eqref{delta-XX} holds.
\end{proof}

Note that for $X\in\Aa$ and $a_i\neq 0$ this proposition is obvious due to the identity $X-\delta^{-1}X^T\delta=2X$.

\section{Isomorphisms of Lie algebras}

In this Section we present several results concerning isomorphisms between Lie algebras $(\Aa,\cmm_S)$ for various choices of parameters $a_1,\ldots,a_{n-1}$ and $S\in\Sa$. In particular in some cases we describe isomorphisms of $\Aa$ with $\so_S$ or $\so(p,q)$.

\begin{prop}\label{prop:iso-d}

If all $a_i\neq0$ then the mapping
\be \Aa\ni X \longmapsto \delta X \in \so_{S\delta^{-1}}(n)\ee
is an isomorphism of Lie algebras
\be (\Aa, \cmm_S) \cong (\so_{S\delta^{-1}}(n), \cmm_{S\delta^{-1}}).\ee
\end{prop}
\begin{proof}
 From condition \eqref{delta-X} it follows that $\delta X$ is an element of $\so_{S\delta^{-1}}(n)$. Direct calculation shows that
$$\delta [X,Y]_S= [\delta X, \delta Y]_{S\delta^{-1}}.$$
Obviously, since $\delta$ is an invertible matrix, the mapping is invertible.
\end{proof}

Observe that if the first parameter of deformation of Lie algebra $\Aa$ vanishes $a_1=0$, but all other does not, we can introduce the map
\be\A_{0,a_2,\ldots,a_{n-1}}
\ni X \longmapsto \frac12(\delta_1 X - X^T\delta_1)\in \so_{S\,\iota(\delta_1)}(n).\ee
It is again an isomorphism of Lie algebras
\be \label{iso:a1}(\A_{0,a_2,\ldots,a_{n-1}}, \cmm_S) \cong (\so_{S\,\iota(\delta_1)}(n), \cmm_{S\,\iota(\delta_1)}).\ee
Let us recall that $\iota(\delta_1)$ is pseudoinverse of $\delta_1$ defined by \eqref{iota}.

\begin{prop}\label{prop:iso-S}
If all parameters $a_i\neq 0$ then the following Lie algebras are isomorphic
\be (\Aa, \cmm_{C^T\delta SC\delta}) \cong (\Aa, \cmm_{S})\ee
for any invertible matrix $C$. The isomorphism is given by
\be \Aa\ni X\longmapsto C\delta X C^T\delta.\ee
\end{prop}
\begin{proof}
Due to the fact that both $C$ and $\delta$ are invertible, the considered map is an isomorphism of vector spaces.

From \eqref{delta-X} it follows that $C\delta X C^T\delta$ belongs to $\Aa$ and it is easy to check that
$$C\delta [X,Y]_{C^T\delta SC\delta} C^T\delta = [C\delta X C^T\delta, C\delta Y C^T\delta]_S.$$
\end{proof}

Note that since $\delta S$ is a symmetric matrix (see \eqref{delta-S}), then for proper choice of $C$, the matrix $C^T\delta SC\delta$ can be made diagonal.

If we admit some parameters equal to zero then the preceding proposition can be reformulated in the following way

\begin{prop}
The following Lie algebras are isomorphic
\be (\Aa, \cmm_{C^T\tilde\delta SC\tilde\delta}) \cong (\Aa, \cmm_{S}).\ee
The isomorphism is given by
\be \Aa\ni X\longmapsto C\tilde\delta X C^T\tilde\delta,\ee
where
\be \tilde\delta=\delta_{k_0}+\delta_{k_1}+\ldots+ \delta_{k_N}\ee
and $C$ is an invertible matrix block-diagonal with respect to the block decomposition \eqref{schodki}, i.e.
\be P_{k_i} C (\1-P_{k_i})=(\1-P_{k_i})C P_{k_i}=0, \qquad i=0,\ldots,N\ee

\end{prop}

In this case in general it is not possible to choose $C$ such that $C^T\tilde\delta SC\tilde\delta$ will be diagonal. However we can always make blocks on diagonal (see \eqref{block-schodki}, i.e. $(P_{k_i}-P_{k_{i+1}})C^T\tilde\delta SC\tilde\delta (P_{k_i}-P_{k_{i+1}})$) to be diagonal matrices.

\begin{prop}\label{prop:so_pq}
 If all $a_i\neq0$ and $S$ is invertible then $\Aa$ is a Lie algebra isomorphic to $\so(p,q)$ for some numbers $p,q$.
\end{prop}
\begin{proof}
Composing the isomorphisms from Propositions \ref{prop:iso-d} and \ref{prop:iso-S} it follows that $\Aa$ is isomorphic to $\so_{C^TS\delta^{-1}C}(n)$ 
for any invertible matrix $C$. Since $S\delta^{-1}$ is invertible, by Sylvester's theorem we see that by proper choice of $C$ the matrix $C^TS\delta^{-1}C$  can be diagonalized with only $\pm1$ on the diagonal. Again using Proposition \ref{prop:iso-d} we observe that considered Lie algebra is indeed isomorphic with $\so(p,q)$.
\end{proof}

\begin{prop}\label{prop:semidirect}
 If $a_k=0$ for some $k$, then $\Aa$ is a semidirect product
of $\A_{a_1,\ldots,a_{k-1}}\times\A_{a_{k+1},\ldots,a_{n-1}}$ by
$Mat_{(n-k)\times k}(\R)$, where we consider the set of $(n-k)\times k$ matrices as an abelian Lie algebra and the action of $\A_{a_1,\ldots,a_{k-1}}\times\A_{a_{k+1},\ldots,a_{n-1}}$ on $Mat_{(n-k)\times k}(\R)$ is given by
$$(A,C)\cdot B := CB-BA.$$
\end{prop}
\begin{proof}
The equality of $\Aa$ and
$(\A_{a_1,\ldots,a_{k-1}}\times\A_{a_{k+1},\ldots,a_{n-1}})\ltimes Mat_{(n-k)\times k}(\R)$ as vector spaces follows from conditions \eqref{delta-X2} and block representation \eqref{block-tr}.

If we identify a matrix $X\in \Aa$ in the block form \eqref{block-tr} with a triple $(A,C,B)$ and calculate Lie bracket, we get
$$[(A,C,B),(A',C',B')]=\big([A,A'],[C,C'],BA'+CB'-B'A-C'B\big)=$$
$$=\big([A,A'],[C,C'],(A,C)\cdot B' -(A',C')\cdot B\big),$$
which is a Lie bracket for a semidirect product.
\end{proof}

\section{Lie--Poisson structure on dual space to $\Aa$}
The dual space to the Lie algebra $\Aa$ can be identified with the space of upper-triangular matrices
\be L_+:=\{\rho=(\rho_{ij})\in Mat_{n\times n}(\R)\;|\;
\rho_{ij}=0\textrm{ for }i\geq j\}\ee
using a pairing given by the trace
\be \pair\rho X=\Tr(\rho X),\ee
where $\rho\in L_+$ and $X\in\Aa$, see \cite{DO-L2}. Other realizations of $(\Aa)^*$ are also interesting but in this paper we restrict our attention to one choice.

The coadjoint representation of $\Aa$ on $L_+$ can be expressed in the following form
\be\label{coad} \ad^*_X\rho =\pi (\rho X S - SX\rho),\ee
where 
\be \pi (A)= \pi^+\big(A-\sum_{i=0}^N \iota(\delta_{k_i}) A^T \delta_{k_i})\big)\ee
and $\pi^+$ is a truncation to strictly upper-triangular matrix. Note that when $S$ is an invertible matrix,  we can express $\ad^*$ in the following form
\be  \ad^*_X\rho=\pi([\rho, S X S]_{S^{-1}}).\ee

The duality $(\Aa)^*\cong L_+$ defines a family of canonical Lie--Poisson brackets on $L_+$ indexed by $S\in\Sa$ in the following way
\be \label{pb-S}\{f,g\}_S(\rho)=\pair\rho{[Df(\rho),Dg(\rho)]_S},\ee
where $f,g\in C^\infty(L_+)$, $\rho\in L_+$. The derivative $Df(\rho)$ here is considered as an element of $\Aa$.

From the Lie bundle structure on $\Aa$ we obtain the following proposition. 

\begin{prop}
 The Poisson brackets $\pb_{S_1}$ and $\pb_{S_2}$ for $S_1,S_2\in\Sa$ are compatible, i.e. their linear combinations are again Poisson brackets on $L_+$. Moreover
\be \label{pencil}\{f,g\}_{S_1}+\lambda \{f,g\}_{S_2} = \{f,g\}_{S_1+\lambda S_2}\ee
for all $f,g\in C^\infty(L_+)$ and $\lambda\in\R$.
\end{prop}
\begin{proof}
 The formula \eqref{pencil} follows directly from the definition of Lie--Poisson brackets \eqref{pb-S} and Lie bundle structure on $\Aa$ \eqref{Sbracket}:
$$
\{f,g\}_{S_1}(\rho)+\lambda \{f,g\}_{S_2}(\rho)=\Tr(\rho[Df,Dg]_{S_1})+\lambda\Tr(\rho[Df,Dg]_{S_2})=$$
$$=\Tr(\rho([Df,Dg]_{S_1}+\lambda [Df,Dg]_{S_2})=\Tr(\rho([Df,Dg]_{S_1+\lambda S_2})=\{f,g\}_{S_1+\lambda S_2}(\rho).
$$
\end{proof}

\begin{prop}
For the case $a_i\neq 0$ and $\det S\neq 0$ the functions
\be \label{casimirs}C_l(\rho)=\frac1{2l}\Tr((\rho-\delta^{-1}\rho^T\delta)S^{-1})^{2l}\ee
are global Casimirs for Lie--Poisson bracket \eqref{pb-S}.
\end{prop}
\begin{proof}
Let us calculate the derivative of the function $C_l$ from the definition:
$$C_l(\rho+\Delta\rho)-C_l(\rho)=\frac1{2l}\Tr(\rho S^{-1} - \delta^{-1}\rho^T\delta S^{-1} + \Delta\rho S^{-1} - \delta^{-1}\Delta\rho^T\delta S^{-1})^{2l} - $$
$$-\frac1{2l}\Tr(\rho S^{-1} - \delta^{-1}\rho^T\delta S^{-1})^{2l}.$$
Reordering the terms under the trace and taking into account only terms linear in $\Delta\rho$ we get
$$C_l(\rho+\Delta\rho)-C_l(\rho)=
\Tr\big( (\rho S^{-1} - \delta^{-1}\rho^T\delta S^{-1})^{2l-1}(\Delta\rho S^{-1}-\delta^{-1}\Delta\rho^T\delta S^{-1})\big)+O(\Delta\rho^2)=$$
$$=\Tr\big(S^{-1}(\rho S^{-1} - \delta^{-1}\rho^T\delta S^{-1})^{2l-1}\Delta\rho - \delta S^{-1}
(\rho S^{-1} - \delta^{-1}\rho^T\delta S^{-1})^{2l-1}\delta^{-1}\Delta\rho^T\big)+O(\Delta\rho^2).$$
Thus we conclude
$$ DC_l(\rho)=S^{-1}\big((\rho-\delta^{-1}\rho^T\delta)S^{-1}\big)^{2l-1}-
\delta^{-1}\big(S^{-1}\big((\rho-\delta^{-1}\rho^T\delta)S^{-1}\big)^{2l-1}\big)^T\delta.$$
From Proposition \ref{prop:4} it follows that $(\rho-\delta^{-1}\rho^T\delta)\in \Aa$. Subsequently from Proposition \ref{prop:2} and \ref{prop:5} we get that
$$ S^{-1}\big((\rho-\delta^{-1}\rho^T\delta)S^{-1}\big)^{2l-1}\in\Aa.$$
Again according to the Proposition \ref{prop:4} (and comment after it) we see that obtained expression can be simplified to
$$ DC_l(\rho)=2S^{-1}\big((\rho-\delta^{-1}\rho^T\delta)S^{-1}\big)^{2l-1}\in\Aa.$$

Let us now compute the Poisson bracket of $C_l$ with an arbitrary function $f\in C^\infty(L_+)$:
$$ \{C_l,f\}_S(\rho)=\Tr \big(\rho[DC_l(\rho),Df(\rho)]_S\big)=$$
$$=2Tr\bigg(\rho(S^{-1}(\rho-\delta^{-1}\rho^T\delta))^{2l-1}-((\rho-\delta^{-1}\rho^T\delta)S^{-1})^{2l-1}\rho\bigg)Df(\rho)=$$
$$=2\Tr\bigg(-\rho S^{-1}((\rho-\delta^{-1}\rho^T\delta)S^{-1})^{2l-2}\delta^{-1}\rho^T\delta+
\delta^{-1}\rho^T\delta S^{-1}((\rho-\delta^{-1}\rho^T\delta)S^{-1})^{2l-2}\rho\bigg)Df(\rho).$$

We will now check that
$$\delta^{-1}\rho^T\delta S^{-1}((\rho-\delta^{-1}\rho^T\delta)S^{-1})^{2l-2}\rho-
\rho S^{-1}((\rho-\delta^{-1}\rho^T\delta)S^{-1})^{2l-2}\delta^{-1}\rho^T\delta
\in\Sa.$$
To improve readability let us denote $S^{-1}(\rho-\delta^{-1}\rho^T\delta)S^{-1})^{2l-2}$ by $W$. Note that by Proposition \ref{prop:2} the matrix $W$ belongs to $\Sa$. Thus we have to check that
$$ \delta(\delta^{-1}\rho^T\delta W \rho - \rho W \delta^{-1}\rho^T\delta)-
(\delta^{-1}\rho^T\delta W \rho - \rho W \delta^{-1}\rho^T\delta)^T\delta=0.$$
It follows in a straightforward way:
$$ \rho^T\delta W\rho -\delta \rho W\delta^{-1}\rho^T\delta-\rho^TW^T\delta\rho+
\delta\rho \delta^{-1} W^T \rho^T\delta=$$
$$=\rho^T\delta W\rho-\delta \rho \delta^{-1} W^T\rho^T\delta-\rho^T\delta W\rho+
\delta\rho \delta^{-1} W^T \rho^T\delta=0.$$

Now using Proposition \ref{prop:3} we conclude that
$$ \{C_l,f\}_S(\rho)= 0$$
for all $f\in C^\infty(L_+)$.
\end{proof}

The formula \eqref{casimirs} has no sense in the case when one or more parameters $a_k$ are equal to zero. However by taking limit in a certain way we can recover some Casimirs from that formula. Namely let us consider functions 
\be \tilde C_l(\rho)=(a_1\cdots a_{n-1})^lC_l(\rho) = \frac1{2l}\Tr \big( a_1\cdots a_{n-1} (\rho S^{-1})^2 - \eta (S^{-1}\rho)^T\delta\rho S^{-1}-\ee
$$-\rho S^{-1}\eta (S^{-1}\rho)^T\delta  + \eta ((S^{-1}\rho)^T)^2\delta\big)^l,$$
where $\eta=a_1\cdots a_{n-1}\, \delta^{-1}$. Note that now each term on the right hand side makes sense in the case even when parameters tend to zero. Moreover in that limit the expression simplifies due to the fact that $\eta\delta\to0$. Thus we obtain
\be \label{casimir_zero}\tilde C_l(\rho)=\frac{(-2)^l}{2l}\Tr\big(\rho S^{-1}\eta(S^{-1}\rho)^T\delta\big)^l.\ee
Direct check confirms that $\tilde C_l$ are Casimirs in the considered case. In general however this family has less functionally independent Casimirs than in non-degenerate case due to the fact that both $\eta$ and $\delta$ have nontrivial kernels. Thus the common level set for those Casimirs may not yet be a symplectic leaf.

It is a well known fact (see e.g. \cite{magri,adler}) that Casimirs for some bracket from a nontrivial Poisson pencil (i.e. when brackets are not proportional) are functions in involution with respect to all other Poisson brackets from that pencil. We can fix one Lie--Poisson bracket $\pb_S$ and consider a hierarchy of Hamilton equations generated by Casimirs for Lie--Poisson bracket $\pb_{\tilde S}$:
\be \frac{d\rho_{ij}}{dt_l}=\{C_l,\rho_{ij}\}_S\ee
or alternatively
\be\label{ham-ad}\frac{d\rho}{dt_l}=\ad^*_{DC_l}\rho.\ee
The flows of equations of this hierarchy commute with each other and thus the equations have a rich set of integrals of motion in involution.

In our case, we can rewrite Casimirs \eqref{casimirs} for the Poisson pencil $\pb_{\tilde S}$, where $\tilde S=S_1+\lambda S_2$
\be \label{casimirs-l}C^{\tilde S}_l(\rho)=\frac1{2l}\Tr((\rho-\delta^{-1}\rho^T\delta)(S_1+\lambda S_2)^{-1})^{2l}.\ee
These Casimirs can be expanded into a power series with respect to $\lambda$. By using power series decomposition with respect to $\lambda$ we get
\be \label{casimir-dec}C^{\tilde S}_l(\rho)=C^{S_1}_l(\rho)+\lambda\Tr\big(((\rho-\delta^{-1}\rho^T\delta)S_1^{-1})^{2l}S_2S_1^{-1}\big)+\ldots.\ee
It is also known that the coefficients of that decomposition are additional integrals of motion in involution.

\section{Examples}

Let us consider $n=4$ case. Thus the considered Lie algebra
$\A_{a_1,a_2,a_3}$ depends only on 3 parameters $a_1$, $a_2$ and $a_3$.

For $S=\1$ the considered class of deformed Lie algebras contains several known classical Lie algebras for the certain choices of the deformation parameters:
\bit
\item $\so(4)$ for $a_1=a_2=a_3=1$;
\item $\so(3,1)$ for $a_1=-1$, $a_2=a_3=1$;
\item $\so(2,2)$ for $a_1=a_3=-1$, $a_2=1$;
\item Euclidean Lie algebra $\mathfrak e(3)$ for $a_1=0$, $a_2=a_3=1$;
\item Poincar\'e Lie algebra $\mathfrak p(1,2)$ for $a_1=0$, $a_2=-1$, $a_3=1$;
\item Galilean Lie algebra $\mathfrak g(2)$ for $a_1=a_2=0$, $a_3=1$;
\item semidirect product $(\so(2)\times\so(2))\ltimes Mat_{2\times2}(\R)$ for $a_1=a_3=1$, $a_2=0$ (see Proposition \ref{prop:semidirect}).
\eit

We will now consider two Lie brackets on $\A_{a_1,a_2,a_3}$ given by two choices of matrix from $\Sa$, 
Taking into account Proposition \ref{prop:iso-S} we restrict our considerations to diagonal $S$:
\be S=\operatorname{diag}(s_1^{-1},s_2^{-1},s_3^{-1},s_4^{-1}).\ee

Let us introduce the notation $\vec x=(\rho_{12},\rho_{13},\rho_{14})$, $\vec y=(\rho_{34},-\rho_{24},\rho_{23})$ and write the matrix $\rho\in L_+$ as follows
\be\rho=
\left(\begin{array}{cccc}
0 & x_1 & x_2 & x_3\\
0 & 0 & y_3 & -y_2\\
0 & 0 & 0 & y_1\\
0 & 0 & 0 & 0
\end{array}\right).
\ee
In order to calculate Lie--Poisson bracket \eqref{pb-S} let us observe that the derivative of a function $f\in C^\infty(L_+)$ assumes the form
\be
Df(\rho)=\left(
\begin{array}{ccccc}
0       & -a_{1}\frac{\partial f}{\partial x_1} & -a_{1}a_2\frac{\partial f}{\partial x_2} & -a_1a_2a_3\frac{\partial f}{\partial x_3} \\
\frac{\partial f}{\partial x_1}  & 0            & -a_{2}\frac{\partial f}{\partial y_3}  & a_2a_3\frac{\partial f}{\partial y_2}\\
\frac{\partial f}{\partial x_2} & \frac{\partial f}{\partial y_3} & 0 & -a_3 \frac{\partial f}{\partial y_1} \\
\frac{\partial f}{\partial x_3} & -\frac{\partial f}{\partial y_2} &\frac{\partial f}{\partial y_1} & 0
\end{array}\right).
\ee
By appropriate grouping of terms we can express the Lie--Poisson bracket \eqref{pb-S} in the following form
\be \label{pb:4}\{f,g\}_S(\rho)= (A\vec x)\cdot \left((\tilde S B \frac{\partial f}{\partial \vec x}) \times \frac{\partial g}{\partial \vec y} -(\tilde SB\frac{\partial g}{\partial \vec x}) \times \frac{\partial f}{\partial \vec y}\right)+\ee
$$+\frac{a_1}{s_1} (A\vec y) \cdot \left(\frac{\partial f}{\partial \vec x} \times \frac{\partial g}{\partial \vec x}\right)
+ (A\tilde SB\vec y) \cdot \left(\frac{\partial f}{\partial \vec y} \times \frac{\partial g}{\partial \vec y}\right),$$
where $A=\operatorname{diag}(a_2,1,1)$, $B=\operatorname{diag}(1,1,a_3)$, $\tilde S=\operatorname{diag}(s_2^{-1},s_3^{-1},s_4^{-1})$.

In this case the Casimirs defined by \eqref{casimirs} assume the following form:
\be C_1(\rho)=-\frac{1}{a_1a_2a_3} (s_1s_2a_2a_3\, x_1^2+s_1s_3a_3\,x_2^2+ s_1s_4\,x_3^2+\ee
$$+s_3s_4a_1a_2\,y_1^2+s_2s_4a_1\,y_2^2+s_2s_3a_1a_3\,y_3^2),$$
\be C_2(\rho)=2C_1(\rho)- \frac{4s_1s_2s_3s_4}{a_1a_2^2a_3}(x_2y_2+x_3y_3+a_2\,x_1y_1)^2.\ee
Note that they are functionally independent. Since $\dim L_+=6$, the generic symplectic leaf is 4-dimensional. Thus we need two functions in involution functionally independent with $C_1$ and $C_2$ (one as a Hamiltonian, the other as integral of motion) to obtain a completely integrable Hamiltonian system. 

Let us take as a Hamiltonian $H$ a function $C^W_1$ (up to a constant factor) computed for a second Poisson bracket $\pb_W$, where $W=\operatorname{diag}(w_1^{-1},w_2^{-1},w_3^{-1},w_4^{-1})$:
\be H(\rho)=\frac{1}{2} (w_1w_2a_2a_3\, x_1^2+w_1w_3a_3\,x_2^2+ w_1w_4\,x_3^2+\ee
$$+w_3w_4a_1a_2\,y_1^2+w_2w_4a_1\,y_2^2+w_2w_3a_1a_3\,y_3^2).$$
The equations of motion for this Hamiltonian assume the form
\be\label{xy-eq}
\left\{ \begin{array}{l}
 \frac{d}{dt}\vec x = \tilde\delta (\tilde W^{-1}\tilde S - \frac{w_1}{s_1}) \big((\tilde W^{-1} \vec x) \times (B\vec y)\big)\\
 \frac{d}{dt}\vec y = a_2 a_3 A^{-1} B^{-1}\big(w_1  \vec x \times (\tilde W^{-1}\tilde S \vec x) - a_1 \tilde W^{-1}(\vec y\times \tilde W^{-1}\tilde S \vec y)\big),
 \end{array}\right.\ee
where $\tilde\delta=\operatorname{diag}(a_1,a_1a_2,a_1a_2a_3)$, $\tilde W=\operatorname{diag}(w_2^{-1},w_3^{-1},w_4^{-1})$.
These equations can be rewritten in a matrix form in the following way (see \eqref{coad} and \eqref{ham-ad})
\be \frac{d}{dt}(\rho-\delta^{-1}\rho^T\delta)=a_1a_2a_3[W^{-1}S(\rho-\delta^{-1}\rho^T\delta)W^{-1}S,(\rho-\delta^{-1}\rho^T\delta)]_{S^{-1}}.\ee
It is an Euler-like equation on Lie algebra $\Aa$ (see \cite{arnold,manakov}) which can be viewed as a deformation of $6$-dimensional rigid body. Moreover, $3$-dimensional rigid body is contained in the system of equations \eqref{xy-eq} and can be obtained by putting $\vec x = 0$.

As an additional integral of motion one might take a first term of the decomposition of $C^{\tilde S}_1$ with respect to $\lambda$, see \eqref{casimir-dec} (up to the constant)
\be I(\rho)=(\frac{s_1}{w_1}+\frac{s_2}{w_2})s_1s_2a_2a_3 x_1^2+
(\frac{s_1}{w_1}+\frac{s_3}{w_3})s_1s_3a_3 x_2^2+
(\frac{s_1}{w_1}+\frac{s_4}{w_4})s_1s_4 x_3^2+\ee
$$+(\frac{s_3}{w_3}+\frac{s_4}{w_4})s_3s_4a_1a_2 y_1^2+
(\frac{s_2}{w_2}+\frac{s_4}{w_4})s_2s_4a_1 y_2^2+
(\frac{s_2}{w_2}+\frac{s_3}{w_3})s_2s_3a_1a_3 y_3^2.$$

Let us now write Hamilton equation for one of the contractions of $\Aa$. Note that from the form of Poisson bracket \eqref{pb:4} we see that putting $a_1=0$ is analogous to putting $1/s_1=0$ (see also isomorphism \eqref{iso:a1}). Let us now consider a Poisson pencil with $W=\operatorname{diag}(0,w_2^{-1},w_3^{-1},w_4^{-1})$ and take as a Hamiltonian a Casimir for $\pb_W$:
\be H(\rho)=\frac{1}{2} (w_2a_2a_3\, x_1^2+w_3a_3\,x_2^2+ w_4\,x_3^2).\ee
Hamilton equations in this case assume the form
\be
\left\{ \begin{array}{l}
 \frac{d}{dt}\vec x = - \frac{1}{s_1}\tilde\delta  \big((\tilde W^{-1} \vec x) \times (B\vec y)\big)\\
 \frac{d}{dt}\vec y = a_2 a_3 A^{-1} B^{-1}\big(\vec x \times (\tilde W^{-1}\tilde S \vec x)\big)
 \end{array}\right.\ee
and are an analogue of Clebsch system, see \cite{clebsch,komarov-tsiganov,yan,dragovic}.

\newcommand{\etalchar}[1]{$^{#1}$}

\end{document}